\documentclass{eptcs}
\usepackage{amsmath,amssymb}
\usepackage{tikz}
\usepackage{graphicx}
\usepackage[utf8]{inputenc}
\usepackage[T1]{fontenc}
\usepackage{url}
\usepackage{quiver}
\usepackage{enumitem}
\usepackage{setspace}
\usepackage{amsthm}
\newtheorem{theorem}{Theorem}
\newtheorem{proposition}[theorem]{Proposition}
\newtheorem{lemma}[theorem]{Lemma}
\newtheorem{corollary}[theorem]{Corollary}
\newtheorem{example}[theorem]{Example} 
\newtheorem{definition}[theorem]{Definition} 
\newcommand{\fillBox}{\hfill$\Box$}
\usepackage{dmcore3}
\usepackage{dmmath3}
\usepackage{dmtproof3}

\newcommand{\fig}{Figure}

\dmNewKeyRing{logic}{#1}
\logicAddIdents{%
  {CL}{IL}{BI}{BBI}
  {K}{T}{4}{5}{S4}{S5}{IS4}{IS5}%
  {BL}{SL}
}
\newcommand{\BL}{\logic{BL}}

\dmNewKeyRing{df}{#1}

\dfAddIdents{,.;()[]=/+-<>0123456789}
\dfAddTuples{%
  {/=}{\neq}
  {>=}{\geqslant}
  {<=}{\leqslant}
  {...}{\ldots}
  {[[}{[}
  {]]}{]}
  {--}{\cdot} 
  {|}{\mid}
  {in}{\in}:{\mathrel{#1}}
  {/in}{\not\in}
  {all}{\forall\,}
  {exists}{\exists\,}
  {st}{\text{such that}}
  {|->}{\mapsto}
  {->}{\rightarrow}
  {-->}{\longrightarrow}
  {=>}{\Rightarrow}
  {==>}{\Longrightarrow}
  {::=}{::=}
  {if.}{\text{if }}
  {.if.}{\text{ if }}
  {.then.}{\text{ then }}
  {then.}{\text{then }}
  {iff.}{\text{ iff}}
  {.iff.}{\text{ iff }}
  {.and.}{\text{ and }}
  {.or.}{\text{ or }}
  {.implies.}{\text{ implies }}
  {.imply.}{\text{ imply }}
  {such.}{\text{such that }}
  {.such.}{\text{ such that }}
  {forall.}{\text{for all }}
  {.forall.}{\text{ for all }}
  {forsome.}{\text{for some }}
  {.forsome.}{\text{ for some }}
}
\dfAddTuples{%
  {S}{\mathrm{S}} 
  {emptyset}{\emptyset}
  {union}{\cup}
  {inter}{\cap}
  {except}{-}
  {subset}{\subset}
  {subseteq}{\subseteq}
  {/subset}{\not\subset}
  {/subseteq}{\not\subseteq}
  {disj}{\mathrm{\#}}
}
\dfAddTuples{%
  {Nats}{\mathbb{N}}
  {Nats*}{{\mathbb{N}^*}}
}
\newcommand{\dfItrFmt}{\dmId}
\dmParseL{#1#2{\dmAddQTKey{df}{\dfItrFmt}{#1}{#2}}}{jjkk}
\dmMapL{\dfAdd{#1}{\logic{#1}}}{%
  {CL}{IL}{SL}{LL}{BL}%
}
\dfAddTuples{%
  {Props}{P} 
  {Forms}{\Phi} 
  {true}{\top} 
  {false}{\bot} 
  {not}{\neg} 
  {and}{\wedge} 
  {or}{\vee} 
  {imp}{\mathbin{\rightarrow}} 
  {eqv}{\mathbin{\leftrightarrow}} 
  {==}{\mathbin{\approx}} 
  {*G}{\otimes} 
  {*}{\mathbin{\ast}}
  {wand}{\mathbin{-\mkern-5mu*}}
  {mimp}{\df{wand}}
  {star}{\ast}
  {mand}{\df{star}}
  {mtrue}{\mathrm{I}}
  {diam}{\Diamond}
  {boxm}{\Box}
}
\newcommand{\dfFrmFmt}{\dmId}
\dmParseL{#1#2{\dmAddQTKey{df}{\dfFrmFmt}{#1}{#2}}}{ppqqrr}
\dmParseL{#1#2{\dmAddQTKey{df}{\dfFrmFmt}{#1}{#2}}}{A{\phi}B{\psi}C{\varphi}}

\dfAddTuples{%
  {Wlds}{W} 
  {FR}{\mathcal{F}}  
  {MO}{\mathcal{M}}  
}
\newcommand{\dfWldFmt}{\dmId}
\dmAddQTKey{df}{\dfWldFmt}{wld}{w}
\dmAddQKeys{df}{\dfWldFmt}{w}
\dfAddTuples{%
  {<=S}{\leq} 
  {/<=S}{\nleq} 
  {=S}{=} 
  {/=S}{\ne} 
  {><S}{\bot} 
  {RelS}{R} 
}
\dfAddTuples{%
  {||-}{\mathbin{\Vdash}} 
  {/||-}{\mathbin{\nVdash}} 
  {|=}{\mathbin{\vDash}} 
  {/|=}{\mathbin{\nvDash}} 
}

\dfAdd{Labs}{L} 
\dfAdd{LAlg}{\mathcal{L}} 
\newcommand{\dfLabFmt}{\dmId}
\dmAddQKeys{df}{\dfLabFmt}{abcd}
\dmAddQKeys{df}{\dfLabFmt}{lxyzuv}
\dfAddTuples{
  {=L}{\!\sim\!} 
  {/=L}{\!\not\sim\!} 
  {<=L}{\!\prec\!} 
  {/<=L}{\!\nprec\!} 
  {><L}{\!\parallel\!} 
}
\makeatletter
\dmDefL\lf#1{\lf@parse#1:\dmEOP\dmEOL}
\dmDefL\lf@parse#1:{\df{#1}\@ifnextchar\dmEOP\lf@nolab{\lf@label}}
\dmDefL\lf@nolab\dmEOP\dmEOL{}
\dmDefL\lf@label#1:\dmEOP\dmEOL{\dmIfEmpty{#1}{}{\mathop{:}\df{#1}}}
\makeatother

\dmNewKeyRing{sg}{#1}
\dmMapL{\sgAdd{#1}{#1}\dfAdd{sg#1}{#1}:{\sgDefaultFmt}}{SFTAR}
\sgAdd{Sopp}{\bar{S}}
\makeatletter
\dmDefL\slf#1{\slf@parse#1\dmEOL}
\dmDefL\slf@parse#1 #2\dmEOL{\sg{#1}\,\lf{#2}}
\dmDefL\sdf#1{\slf@parse#1\dmEOL}
\dmDefL\sdf@parse#1 #2\dmEOL{\sg{#1}\,\df{#2}}
\dmDefL\sgf#1{\slf@parse#1\dmEOL}
\dmDefL\sgf@parse#1 #2\dmEOL{\sg{#1}\,#2}
\makeatother
\dfAdd{Sg}{Sg}

\dmNewKeyRing{ru}{\ensuremath{#1}}
\dmMapL{\ruAdd{T#1}{\sg{T}\df{#1}}\ruAdd{F#1}{\sg{F}\df{#1}}}{%
  {wand}{star}{and}{or}{not}{imp}{eqv}{boxm}{diam}%
}
\ruAddTuples{%
  {CaseD}{CD} 
  {Refl}{\df{<=L}\,R} 
  {Tran}{\df{<=L}\,T} 
  {Equ}{\df{=L}} 
  {SEqu}{\sg{S}\df{=L}} 
  {Sep}{\df{><L}} 
  {SepC}{\ruVal{Sep}C} 
  {SepP}{\ruVal{Sep}P} 
  {Bif}{B} 
  {XForm}{\times_{\df{Forms}}} 
  {XCons}{\times_{\df{<=L}}} 
  {XSepC}{\times_{\ruVal{Sep}}} 
  {XBifA}{\times_{\ruVal{Bif}A}} 
  {XBifS}{\times_{\ruVal{Bif}S}} 
  {Ftrue}{\times_{\df{true}}} 
  {Tfalse}{\times_{\df{false}}} 
}

\dfAdd{TBL}{T_{BL}}:{#1}
\newcommand{\TBL}{\df{TBL}}
\dfAdd{|-}{\mathbin{\vdash}}
\dfAdd{/|-}{\mathbin{\nvdash}}
\dfAdd{|-TBL}{\mathbin{\vdash}}
\dfAdd{/|-TBL}{\mathbin{\nvdash}}
\dmAddQTKey{df}{\dfFrmFmt}{tb}{t}
\dmAddQTKey{df}{\dfFrmFmt}{tbb}{b}
\newcommand{\tb}{\df{tb}}
\newcommand{\tbb}{\df{tbb}}
\renewcommand{\tsg}[1]{\sg{#1}}
\renewcommand{\tlf}[1]{\lf{#1}}
\renewcommand{\tlc}[1]{\df{#1}}

\newcommand{\ecl}[2][]{\dmSub{\df{[ #2 ]}}{\df{#1}}}
\dfAddTuples{%
    {ecx}{\ecl{x}}
    {ecy}{\ecl{y}}
    {ecz}{\ecl{z}}
    {ecu}{\ecl{u}}
    {ecv}{\ecl{v}}
}
\newcommand{\ter}[3]{\df{RelS}(#1,#2,#3)}
\newcommand{\terdf}[3]{\df{RelS}(\df{#1},\df{#2},\df{#3})}
\newcommand{\domFun}{D}
\newcommand{\dom}[1]{\domFun(#1)}
\newcommand{\realFun}{\rho}
\newcommand{\real}[1]{\realFun(\df{#1})}
\dmNewBin{#1#2#3{#2#1#3}}{eqL}{\df{=L}}{#1}{\df{#1}}{\df{#1}}
\dmNewBin{#1#2#3{#2#1#3}}{eqS}{\df{=S}}{#1}{\df{#1}}{\df{#1}}
\dmNewBin{#1#2#3{#2#1#3}}{leqL}{\df{<=L}}{#1}{\df{#1}}{\df{#1}}
\dmNewBin{#1#2#3{#2#1#3}}{leqS}{\df{<=S}}{#1}{\df{#1}}{\df{#1}}
\dmNewBin{#1#2#3{#2#1#3}}{sepL}{\df{><L}}{#1}{\df{#1}}{\df{#1}}
\dmNewBin{#1#2#3{#2#1#3}}{sepS}{\df{><S}}{#1}{\df{#1}}{\df{#1}}
\newcommand{\bifL}[3]{\leqL{#1}{\sepL{#2}{#3}}}
\newcommand{\bifS}{\terdf}
\dfAdd{ValS}{V}
\dmNewUna{#1#2{#1(#2)}}{ValS}{\df{ValS}}{#1}{\df{#1}}
\newcommand{\mimp}{\df{wand}}

\title{Bifurcation Logic: Separation Through Ordering}

\author{Didier Galmiche
\institute{Universit\'e de Lorraine, CNRS, LORIA\\
F-54000 Nancy, France}
\email{didier.galmiche@loria.fr}
\and
Timo Lang
\institute{University College London\\ London, UK}
\email{timo.lang@ucl.ac.uk} 
\and
Daniel M\'ery
\institute{Universit\'e de Lorraine, CNRS, LORIA\\
F-54000 Nancy, France}
\email{daniel.mery@loria.fr}
\and
David Pym 
\institute{UCL and Institute of Philosophy\\
University of London, UK}
\email{d.pym@ucl.ac.uk, david.pym@sas.ac.uk}
}
\def\titlerunning{Bifurcation Logic}
\def\authorrunning{D. Galmiche, T. Lang, D. M\'ery, and D. Pym}

\begin{document}

\maketitle  

\begin{abstract}
We introduce Bifurcation Logic, ${\BL}$, which combines a basic classical modality with separating conjunction, $\ast$, together with its naturally associated multiplicative implication, $\mimp$, that is defined using the modal ordering. Specifically, a formula $\phi_1 \ast\, \phi_2$ is true at a world $w$ if and only if each $\phi_i$ holds at worlds $w_i$ that are each above $w$, on separate branches of the order, and have no common upper bound. We provide a labelled tableaux calculus for $BL$ and establish soundness and completeness relative to its relational semantics. The standard finite model property fails for $BL$. However, we show that, in the absence of $\mimp$, but in the presence of $\ast$, every model has an equivalent finite representation and that this is sufficient to obtain decidability. We illustrate the use of ${\BL}$ through 
an example of modelling multi-agent access control that is quite generic in its form, suggesting many applications.
\end{abstract}

\section{Introduction} \label{sec:introduction}

We introduce Bifurcation Logic, ${\BL}$, which combines a basic classical modality with separating conjunction, $\ast$, together with its naturally associated multiplicative implication, $\mimp$, that is defined using the modal ordering. We provide a relational semantics and a labelled tableaux calculus for ${\BL}$ and establish soundness and completeness relative to ${\BL}$'s relational semantics. The standard finite model property fails for $BL$. However, we show that, in the absence of $\mimp$, but in the presence of $\ast$, every model has a finite representation and that this is sufficient to obtain decidability. We illustrate the use of ${\BL}$ in logical modelling through 
an example of multi-agent access control. 


The key property of interest in ${\BL}$ is the semantics of its multiplicatives, the `separating' $\ast$ and $\mimp$, which is given in terms of the ordering that is used to define the classical modality. This stands in contrast to the set-up in, say, bunched implications (${BI}$ --- e.g., \cite{OP99,POY2004,GMP2005}) in which specific relational structure is used for their definition --- see \cite{GheorghiuPym2023} for a through discussion of $BI$'s semantics. 
Kamide's account of Kripke semantics for modal substructural logics \cite{Kamide2002} also employs a binary operation on worlds to give a treatment of the multiplicative conjunction that is similar to that  of $BI$'s elementary semantics (e.g., \cite{OP99,POY2004,GMP2005,GheorghiuPym2023}). Galmiche, Kimmel, and Pym \cite{GKP2020} consider an epistemic modal extension of boolean $BI$ in which the semantics of the multiplicative conjunction employs a monoidal product on worlds. Do\v{s}en \cite{Dosen1992} considers a range of issues in the relationship between modal and substructural logics from the perspective of translations between proof systems, and Ono \cite{Ono2019} has also considered the proof theory of modal and substructural logics. 

The basic idea in $BL$ is that a formula $\phi_1 \ast\, \phi_2$ is true at a world $w$ if and only if each $\phi_i$ holds at worlds $w_i$ that are each above $w$, on separate branches of the order, and have no common upper bound --- that is, they are bifurcated.  Consequently, the semantics of the multiplicative implication has the property that the implicational formula $\phi \mimp \psi$ and its subformula $\phi$ are required to hold at bifurcated worlds above the world at which $\psi$ holds. This use of this feature is illustrated in a substantive modelling example given in Section~\ref{sec:modelling}.  The semantics of the classical connectives and modality is standard.   

In Section~\ref{sec:bifurcation-logic}, we introduce the language of Bifurcation Logic and its models, based on frames with a ternary relation structure for the bifurcation semantics. In Section~\ref{sec:modelling}, we give an extended, quite generic --- i.e., evidently mappable to other settings --- 
example of the use of ${\BL}$ in modelling access control, suggesting wider application in knowledge representation and reasoning. This example, albeit somewhat idealized, illustrates the interaction between the classical modality and the multiplicative connectives, especially the somewhat unusual semantic form of the implication, in a simple and direct way. We also discuss some related work in this section. 

In Section~\ref{sec:tableaux}, we give a system of labelled tableaux for ${\BL}$. The form of the calculus follows the pattern established in, for example,  \cite{GMP2005,GKP2020} and allows, in Section~\ref{sec:SandC}, soundness and completeness results to be established (cf.~\cite{GMP2005,GKP2020}). The proofs are provided in the appendix.  

While the usual finite model property fails for $\BL$, a modified form of it does hold. Section~\ref{sec:FMP} explains, through a counterexample, why the standard form fails and introduces a modified form through the concept of  `model with back links'. Intuitively, ${\BL}$ is a logic with the subformula property --- in the sense evident from the tableaux system --- and is about paths in finitely branching trees. Back links describe how the paths go back to already-seen configurations (see Section~\ref{sec:FMP} for a formal explanation). Using this modified finite model property, the decidability of ${\BL}$ is obtained. 

Finally, before proceeding to our formal development, we consider a few interesting outstanding questions, among many others,  for further work: 
\begin{itemize}[leftmargin=6mm,label=--]\setlength\itemsep{-0.5mm}
    \item we would aim to give a (Hilbert-type) axiomatization of ${\BL}$; 
    \item we would explore natural deduction and sequent calculus presentations of ${\BL}$; 
    \item we would seek to establish the complexity of deciding ${\BL}$; 
    \item we would seek to explore the addition to BL of multi-agent and epistemic modalities, as well as quantifiers, so extending its potential as a modelling tool 
    (cf. 
    \cite{GKP2020}, for example). 
\end{itemize}
The question of whether there is an interesting intuitionistic version of ${\BL}$ seems quite challenging, as it would combine the difficulties of intuitionistic modal logics \cite{Simpson94} with the need to handle the multiplicatives 
in a coherent way.

\section{Bifurcation Logic} \label{sec:bifurcation-logic}

In this section, we introduce Bifurcation Logic, $\BL$, by giving a definition in terms of ternary relational semantics in the style of Routley-Meyer \cite{RM73}, with some similarity to the work of Fuhrmann and Mares  \cite{FuhrmannMaresOnS}.

\begin{definition}[Language]
  Let $\df{Props}$ be a countable set of propositional letters.
  The formulae of~$\BL$, the set of which is denoted by $\df{Forms}$, are given by the following grammar: 
  \[
  \df{ A ::= p \, ( in Props ) | not A | A and A | boxm A | A star A | A wand A }
  \]
  The connectives $\df{or}$, $\df{imp}$, $\df{eqv}$, $\df{diam}$, and the units $\df{true}$, $\df{false}$ are defined as follows:
      $\df{A or B = not ( not A and not B )}$, 
      $\df{A imp B = not A or B}$, 
      $\df{A eqv B = ( A imp B ) and ( B imp A )}$,  
     $\df{true = A or not A }$, 
     $\df{false = A and not A }$, 
     $\df{diam A = not boxm not A}$. \fillBox 
\end{definition}

To minimize the use of parentheses; we use the following strict order of precedence (with right associativity): 
$
  \df{ boxm , diam , not \;>\; star \;>\; and , or \;>\; wand \;>\; imp \;>\; eqv}.
$

\begin{definition}[Tree]\label{def:tree}
  Let $\df{( Wlds , <=S )}$ be a partial order. 
  Two elements~$\df{wld1 , wld2 in Wlds}$ are \emph{separated (or disjoint)}, denoted $\sepS{wld1}{wld2}$, if neither $\df{wld1 <=S wld2}$ nor $\df{wld2 <=S wld1}$.
  We call $\df{( Wlds , <=S )}$ \emph{rooted} if there exists $\df{wld in Wlds}$ such that $\df{ wld <=S wld'}$ for all~$\df{wld' in Wlds}$.
  $\df{( Wlds , <=S )}$ has the \emph{persistent separation property} if it satisfies the following condition:
  \[
    (P) \quad \df{ forall. wld1 , wld2 , wld1' in Wlds , 
       .if. \sepS{wld1}{wld2} .and. wld1 <=S wld1' , .then. \sepS{wld1'}{wld2}%
    }.
  \]
  $\df{( Wlds , <=S )}$ is called a \emph{tree} if it is rooted 
  and has the persistent separation property. 
  \fillBox
\end{definition}

\begin{definition}[Frame]\label{def:frame}
A $\BL$ \emph{frame} is a structure $\df{FR = ( Wlds , <=S , RelS )}$, where $\df{( Wlds , <=S )}$ is a tree of elements called \emph{worlds} and $\df{RelS}$ is the ternary relation on worlds defined as follows:
\[
  (B) \quad \df{forall. wld , wld1 , wld2 in Wlds , \;
            \bifS{wld}{wld1}{wld2} .iff. \leqS{wld}{wld1} , \leqS{wld}{wld2} .and. \sepS{wld1}{wld2}
  }.
\]
That is, $\df{wld1}$ and $\df{wld2}$ belong to distinct futures of $\df{wld}$. 
\fillBox
\end{definition}


 \begin{definition}[Model] \label{def:model}
 A $\BL$ \emph{model} is a triple 
 $\df{MO} = (\df{FR}, \df{ValS}, \df{||-})$,
  where $\df{FR}$ is a $\BL$ frame
  and $\df{ValS}$ is a \emph{valuation function} from $\df{Wlds}$ to $\wp(\df{Props})$.
  The satisfaction relation $\df{||-}$ is inductively defined as
  the smallest relation on $\df{Wlds} \times \df{Forms}$ such that 
\[
\begin{array}{r@{\;\;}c@{\;\;}l}
\df{MO , wld ||- p} & \df{.iff.} & \df{p in \ValS{wld}} \\ 
\df{MO , wld ||- not A} & \df{.iff.} & \df{MO , wld /||- A} \\ 
\df{MO , wld ||- A and B} & \df{.iff.} & \df{MO , wld ||- A .and. MO , wld ||- B} \\ 
\df{MO , wld ||- boxm A} & \df{.iff.} &\df{forall. wld' in Wlds , .if. wld <=S wld' .then. MO , wld' ||- A} \\ 
\df{MO , wld ||- A star B} & \df{.iff.} & \df{forsome. wld1 , wld2 in Wlds .such. \terdf{wld}{wld1}{wld2}}, \\
                          & & \df{MO , wld1 ||- A .and. MO , wld2 ||- B} \\ 
\df{MO , wld ||- A wand B} & \df{.iff.} & \df{forall. wld1 , wld2 in Wlds .such. \terdf{wld2}{wld}{wld1}}, \\  
                          & & \df{if. MO , wld1 ||- A .then. MO , wld2 ||- B}
\end{array}
\]
A formula $\df A$ is \emph{satisfied in a model $\df{MO}$}, denoted $\df{MO ||- A}$, if $\df{MO , wld ||- A}$ for all worlds~$\df{wld}$ in~$\df{MO}$.
We write $\df{wld ||- A}$ instead of $\df{MO , wld ||- A}$ whenever the model is clear from the context.
$\df A$ is \emph{satisfiable} if it is satisfied in some model $\df{MO}$, and \emph{valid}, denoted $\df{||- A}$, if it is satisfied in all models. \fillBox
\end{definition}

\dfAdd{wldA}{\df{wld1}}
\dfAdd{wldI}{\df{wld2}}
\dfSetVal{mtrue}{\mathop{\top^*}}

\dfAdd{loc}{\ell}

\newenvironment{blgraph}[1][]{%
  \begin{tikzpicture}[%
    level distance=3em, 
    edge from parent/.style={draw,->},#1%
   ]
}{%
  \end{tikzpicture}%
}
\newcommand{\blnode}[3]{%
  node (#1) [label=left:$#2$,label=right:$#3$]{$\bullet$}
}

\begin{figure}[ht]
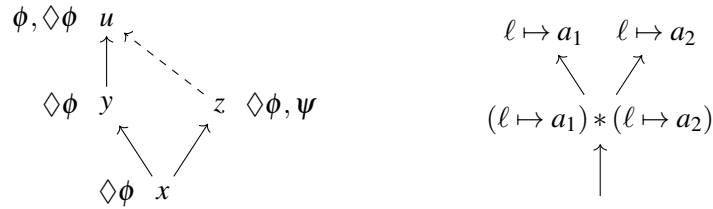

\hrule
\vspace{2mm}
\begin{center}
\begin{tabular}{c@{\hspace{2cm}}c}
  \begin{blgraph}
   \node [label=left:$\df{diam A}$] {$\df{x}$} [grow'=up]
      child {node (y) [label=left:$\df{diam A}$] {$\df{y}$}
          child {node (u) [label=left:$\df{A , diam A}$] {$\df{u}$}}
        }
      child {node (z) [label=right:$\df{diam A , B}$] {$\df{z}$}
      };
      \draw[dashed,->] (z) -- (u);
  \end{blgraph}
  &
  \begin{blgraph}
   \node {} [grow'=up]
           child {  node (x) {$\df{( loc |-> a1 ) star ( loc |-> a2 )}$}
             child { node (y) {$\df{loc |-> a1}$} }
             child { node (z) {$\df{loc |-> a2}$} }
           };
     
  \end{blgraph}
\end{tabular}
\end{center} \vspace{-5mm}
\caption{Examples of $\BL$ structures}
\label{fig:bl-examples}
\vspace{2mm}
\hrule
\end{figure}
$\BL$ uses frames that obey the persistent separation property as we believe that they better correspond to an intuitive understanding of bifurcation.
We then define non-persistent (or lax) $\BL$ as the extension of~$\BL$ that deals with frames that are not required to obey the persistent separation property.
For example, the formula $\varphi = \df{( diam A star B ) eqv ( A star B )}$ is valid in~$\BL$ (a tableau proof is given in Section~\ref{sec:tableaux}), but it is not valid in lax~$\BL$ as it can be falsified in the direct acyclic graph (DAG) given on the left-hand side of \fig~\ref{fig:bl-examples}.
Indeed, we have $\df{x ||- diam A star B}$ because we have $\terdf{x}{y}{z}$, $\df{y ||- diam A}$ and $\df{z ||- B}$, but we do not have $\df{x ||- A star B}$ because the only world satisfying~$\df{A}$ is~$\df{u}$ and~$\df{u}$ does not belong to a distinct future of $\df{z}$, which is the only world satisfying~$\df{B}$.

\begin{proposition}
BL is a conservative extension of the modal logic $S4$. 
\end{proposition}
\begin{proof}
This follows immediately from observing that the only conditions imposed on the order relation are for the purpose of defining $\ast$ and $\mimp$. 
\end{proof}

Although $\BL$ is a conservative extension of $S4$, it differs from both bunched and other separating logics.
First, we remark that since ${\BL}$ addresses separation as an ordering problem rather than a resource composition problem (as in, for example, bunched logics,  where~$\df{star}$ usually corresponds to a product in a monoid), we do not include a unit~$\df{mtrue}$ for the multiplicative conjunction~$\df{star}$.
Having the multiplicative unit~$\df{mtrue}$ would unnecessarily complicate our definition of the bifurcation relation or rule out many partial order structures.
Indeed, we would need to satisfy $\df{wld ||- A * mtrue}$ iff $\df{wld ||- A}$ for all worlds~$\df{wld}$ and all formulae~$\df A$. 
Hence, by definition of $\df{star}$, we would need worlds $\df{wldA}, \df{wldI}$ such that $\terdf{wld}{wldA}{wldI}$, $\df{wldA ||- A}$ and $\df{wldI ||- mtrue}$.
In particular, consider a $\BL$ frame with only one world~$\df{wld}$ and set $\df{A = true}$. We have $\df{wld ||- true}$, but not $\df{wld ||- true star mtrue}$ because $\terdf{wld}{wld}{wld}$ is impossible to achieve as it would imply both $\df{wld <=S wld}$ and $\df{wld /<=S wld}$ by definition of~$\df{RelS}$.

Second, $\BL$ also differs from Separation Logic, as illustrated in the right-hand side of \fig~\ref{fig:bl-examples}. In Separation Logic \cite{Ishtiaq2001,Reynolds2002}, the built-in points-to predicate $\df{( loc |-> a )}$ intuitively denotes a memory heap with only one cell whose location (address) is~$\df{loc}$ and whose value is~$\df{a}$.
Heaps are defined as partial functions from locations to values and composition of heaps is given by the union of functions with disjoint domains.
Therefore, the formula $\df{( loc |-> a1 ) star ( loc |-> a2 )}$ is not satisfiable in Separation Logic as it denotes the disjoint composition of two one-cell heaps that share the same location~$\df{loc}$.
In $\BL$, as~$\df{star}$ represents bifurcation, a node satisfying $\df{( loc |-> a1 ) star ( loc |-> a2 )}$ simply implies that the location~$\df{loc}$ might have two distinct futures, one in which it points to the value~$\df{a1}$ and the other one in which it points to the value~$\df{a2}$.
More interestingly, the formula $\df{ boxm ( loc |-> a ) }$, when satisfied by some world~$\df{wld}$, would imply that in all possible futures of~$\df{wld}$, the location~$\df{loc}$ should point to the same value~$\df{a}$.
This would be useful, for example, to state that an interrupt vector
always points to the address where its legitimate handler resides.

\section{Modelling With Bifurcation Logic} \label{sec:modelling}

Logics can be used not merely to describe reasoning itself, but also to describe reasoning about systems. This use of logics as modelling tools has delivered significant advances in many areas --- too numerous to describe here --- including program analysis and verification, with one leading example, making essential use of multiplicative conjunction, being Separation Logic \cite{Ishtiaq2001,Reynolds2002}. Another highly effective example in the same spirit is Context Logic \cite{CGZ2005,CGZ2007}. More abstractly, substructural modal logics provide reasoning tools for models of systems
in the `distributed systems metaphor'  (e.g., \cite{AP16,CIP2021,GLP2024}). Demri and Deter 
\cite{DD2015} have surveyed connections between modality and separation, but they do not consider separating connectives defined through ordering. More detailed connections with Separation Logic, as  mentioned above, are beyond the scope of this paper;  so, instead of describing how connections with that might work, we give a quite generic example --- many examples involving obtaining, representing, and verifying knowledge and information will be very similar --- of how ${\BL}$'s modality and multiplicatives interact through an example in the context of multi-agent access control. 


\begin{example}[Crimson Tide \cite{CrimsonTide} Part 1] \label{ex-crimson-tide}
The plot of the 
film \emph{Crimson Tide} \cite{CrimsonTide} takes place mainly on board a United States Navy Ballistic Missile Submarine (the USS Alabama). During a period of international tension, the submarine 
receives an order to launch nuclear-armed missiles. For the release of the weapons to be authorized, a composite key must be authenticated:  
\begin{enumerate}
\setlength\itemsep{-0.5mm}
\item An order to release weapons is received in a message that contains a code.   
\item Two senior officers (neither the Captain nor the Executive Officer) must independently authenticate the message by verifying the code against a local authentication device. Authentication (think `necessarily correct') is denoted using $\Box$.  
\item This yields a composite --- $\Box \kappa_1 \ast \Box\kappa_2$ in Figure~\ref{fig:joint-key} --- of authenticated keys that is passed (verbally) to the Executive Officer and the Captain. 
\item The Executive Officer confirms that the correct protocol has been followed and so authenticates the composite key, thereby yielding $\Box(\Box\kappa_1 \ast \Box\kappa_2)$ in Figure~\ref{fig:joint-key}.
\item The Captain then confirms the authentication 
and may order the release of the weapons.
\end{enumerate}
\noindent This concludes the first part of our example. \fillBox

\begin{figure}[t]
\hrule
\vspace{2mm}
    \begin{center}
    \includegraphics[scale=0.25]{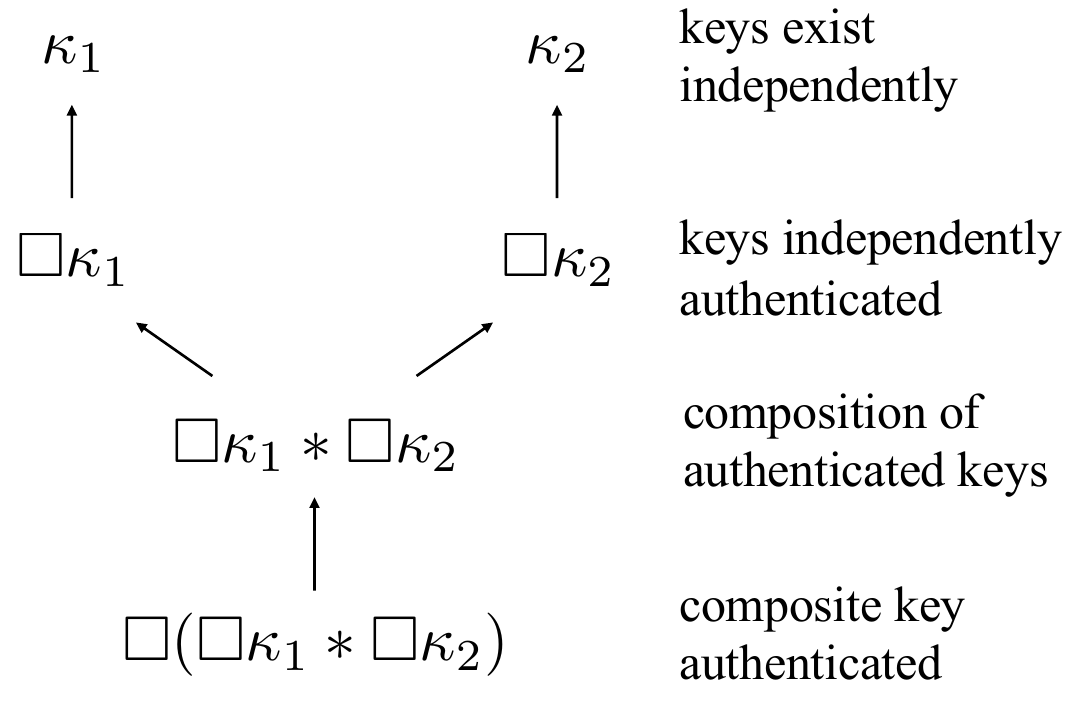}
    \end{center} \vspace{-5mm}
\caption{Joint access}
\label{fig:joint-key}
\vspace{2mm}
\hrule 
\end{figure}
\begin{figure}[h]
\hrule
\vspace{2mm}
    \begin{center}
    \includegraphics[scale=0.25]{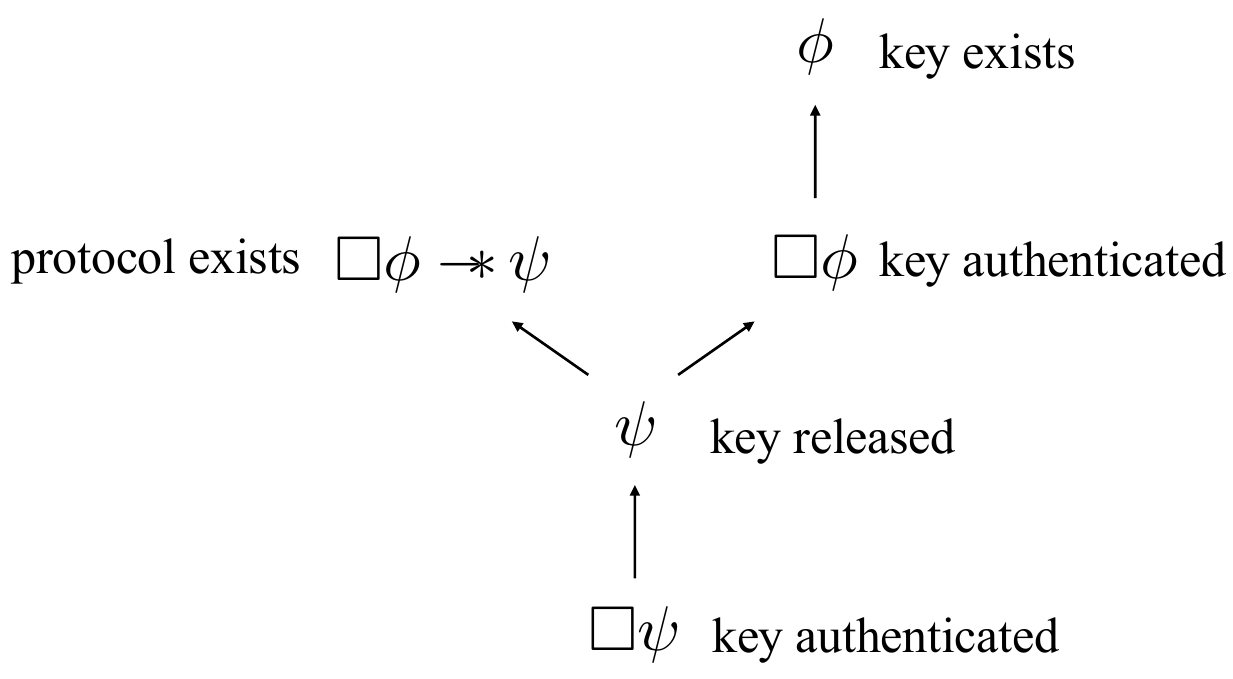}
    \end{center} \vspace{-5mm}
\caption{Protocol for obtaining a key}
\label{fig:protocol}
\vspace{2mm}
\hrule 
\end{figure}
\begin{figure}[ht]
\hrule
\vspace{2mm}
    \begin{center}
    \includegraphics[scale=0.25]{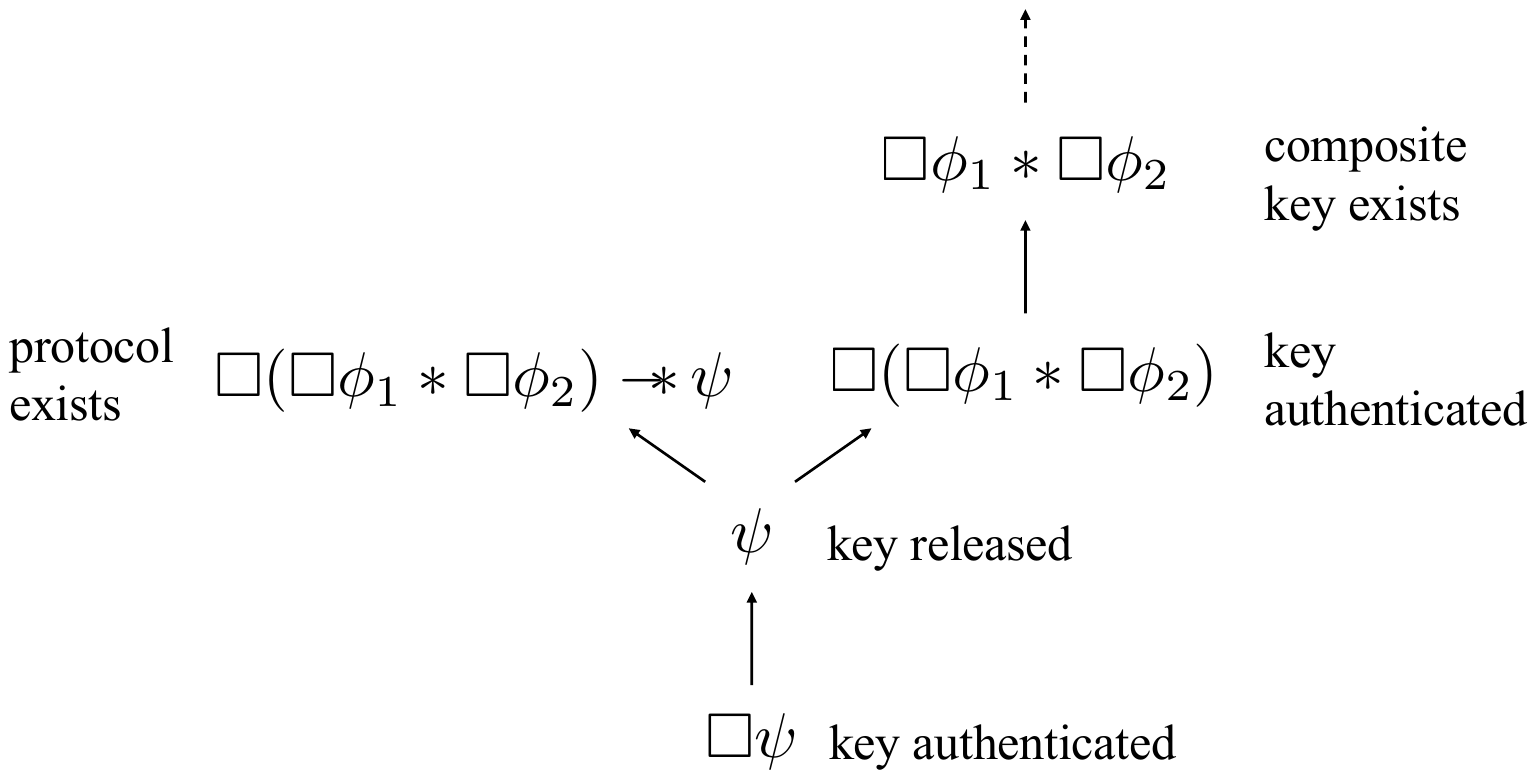}
    \end{center} \vspace{-5mm}
\caption{Protocol for obtaining a composite key}
\label{fig:protocol-comp}
\vspace{2mm}
\hrule
\end{figure}
\end{example}

\begin{example}[Crimson Tide \cite{CrimsonTide} Part 2] 
The use of a protocol can also be represented in Bifurcation Logic. This is illustrated in Figure~\ref{fig:protocol}. Here the idea is that a protocol is modelled by an implicational formula that, given an authenticated key,  yields a key, which may then need to be authenticated. (Informally, the implicational formula may be thought of as the type of a function that returns a key.) 
We can see how in an inessential, slightly more detailed, variant of the set-up,  Figure~\ref{fig:protocol-comp} represents the use of protocols in \emph{Crimson Tide} as follows: 
\begin{enumerate}[leftmargin=6mm]
\setlength\itemsep{-0.7mm}
\item In the discussion based on Figure~\ref{fig:joint-key}, the role of protocols is suppressed. 
\item Figure~\ref{fig:protocol} illustrates the use of protocols in general.  
\item In the case of our example, we represent the Captain's key permitting the release of weapons by $\psi$. For the key to be authenticated, it must have been released through the protocol. 
\item The protocol is given, in Figure~\ref{fig:protocol-comp}, by the implicational formula 
$\Box(\Box\phi_1 \ast \Box\phi_2) \mimp \psi$. That is, 
the authenticated composite key allows access to the Captain's key, $\psi$, which, when authenticated, allows the weapons to be released.  
\end{enumerate}

Here we use the modality $\Box$ to denote authentication, but what is the role of separation? The role of $\ast$ should be clear: it enforces the independence of multiple authentications. But what of $\mimp$? It ensures that there is no interference between the authentication of the composite key and its use to make available the Captain's key, $\psi$ (which must itself be authenticated for use). Note the essential use of the semantics of $\mimp$: first, the separation between worlds afforded by $\mimp$, as opposed to $\df{imp}$, ensures no interference between the existence of the protocol and existence of the (authenticated) key to which it applies. Second, the particular semantic form of this implication --- in that its application `looks back' down the ordering --- captures exactly 
the release of the Captains's key, $\psi$, through access to the 
authenticated composite key. \fillBox
\end{example}


\section[A Tableaux Calculus for BL: TBL]{A Tableaux Calculus for ${\BL}$: $\TBL$} \label{sec:tableaux}

The tableaux calculus for classical propositional logic \cite{Fitt12} can be adapted systematically to calculi for many non-classical logics by the addition of labelling.
The basic idea is that the structure of a Kripke model for a given non-classical logic is used to define a tableaux calculus for that logic by reflecting its structure in an algebra of labels that is used to impose side-conditions on the tableaux rules. Through this mechanism, the basic classical and/or tableaux figures are modified to capture non-classical connectives. 

$\df{TBL}$ is presented in \fig~\ref{fig:tbl-rules}.
$\df{TBL}$ has logical rules that capture the meaning of the connectives, structural rules that capture the properties of $\BL$ models, and closure rules (whose conclusion is a cross mark) that capture (logical or structural) inconsistencies.
As usual, closure rules put an end to the expansion of a branch. 
$\df{TBL}$ can address either ${\BL}$ or its lax variant depending on the inclusion or not of the optional persistency rule $\ru{SepP}$.
All the results for ${\BL}$ presented in this section  also hold for lax ${\BL}$.

\begin{definition}\label{def:labels-and-constraints}
Let $\df{Labs}$ be a countable set of symbols called \emph{labels}. 
A \emph{labelled formula} is a pair $\df{( A , x )}$, written $\lf{A:x}$, where $\df{A}$ is a formula and $\df{x}$ is a label.
A \emph{label constraint} is an expression of the form $\df{x <=L y}$, where $\df{x , y}$ are labels. \fillBox
\end{definition}

\begin{definition}\label{def:signed-formulae}
Let $\df{Sg}$ be the set $\SET{ \sg{T}, \sg{F} }$ of \emph{signs}.
A \emph{signed labelled formula} is a triple $( \sg{S}, \df{A}, \df{x} )$,
written $\slf{S A:x}$, where $\sg{S}$ is a sign and $\lf{A:x}$ is a labelled formula.
Similarly, a \emph{signed label constraint} is a label constraint prefixed with a sign.
\fillBox
\end{definition}

We define $\sdf{T x =L y}$ as a shorthand for the expression $\sdf{T x <=L y}, \sdf{T y <=L x}$.
Similarly, $\sdf{T x ><L y}$ is a shorthand for $\sdf{F x <=L y} , \sdf{F y <=L x}$ and $\sdf{T \bifL{x}{y}{z}}$ is a shorthand for $\sdf{T x <=L y} , \sdf{T x <=L z} , \sdf{T y ><L z}$.


\begin{definition}\label{def:tbl-tableau-for-A} 
  A \emph{tableau for $\lf{A:x}$} is a finitely branching rooted tree built inductively according to the rules given 
  in \fig~\ref{fig:tbl-rules} and the root node of which is the signed labelled formula $\slf{F A:x}$.
  \fillBox
\end{definition}

\begin{definition}\label{def:tbl-closed-tableau}
  A tableau branch~$\tbb$ is \emph{closed} if it ends with a closure rule. 
  A tableau~$\tb$ is \emph{closed} if all of its branches are closed.
  \fillBox
\end{definition}

{
\begin{figure}[tp]
\hrule
\vspace{2mm}

  \newcommand{\tabnl}[1][-1.75ex]{\\[#1]}
  \newcommand{\tabhl}[1][-2ex]{\\[#1]}
  \newcommand{\tabheader}[1]{\hline \tabhl \textbf{#1} \\ \tabhl \hline}
  \centering
  \footnotesize

  \begin{tabular}{|@{\hspace{1ex}}c@{\hspace{1ex}}|}

  \tabheader{Logical rules} \tabnl

  \begin{tproof}
    \tleaf{\tsf{T A:x}}
    \tinfer[tline]1[\ru{Fnot}]{\tsf{F not A:x}}
  \end{tproof}
  \quad
  \begin{tproof}
    \tleaf{\tsf{F A:x}}
    \tinfer[tline]1[\ru{Tnot}]{\tsf{T not A:x}}
  \end{tproof}
  \quad
  \begin{tproof}
    \tleaf{\tsf{T A:x}\tnl\tsf{T B:x}}
    \tinfer[tline]1[\ru{Tand}]{\tsf{T A and B:x}}
  \end{tproof}
  \quad
  \begin{tproof}
    \tleaf{\tsf{F A:x}}
    \tleaf{\tsf{F B:x}}
    \tvrule
    \tinfer[tline]2[\ru{Fand}]{\tsf{F A and B:x}}
  \end{tproof}
  \quad
  \begin{tproof}
    \tleaf{\tsc{T \leqL{x}{u}}\tnl\tsf{F A:u}}
    \tinfer[tline]1[\ru{Fboxm}]{\tsf{F boxm A:x}}
  \end{tproof}
  \quad
  \begin{tproof}
    \tleaf{\tsf{T A:y}}%
    \tinfer[tline]1[\ru{Tboxm}]{\tsf{T boxm A:x}\tnl\tsc{T \leqL{x}{y}}}
  \end{tproof}

  \\ \tabnl

  \begin{tproof}
    \tleaf{%
      \tsc{T \bifL{x}{u}{v}}%
      \tnl%
      \tsf{T A:u}
      \tnl
      \tsf{T B:v}%
    }
    \tinfer[tline]1[\ru{Tstar}]{\tsf{T A star B:x}%
    }
  \end{tproof}
  \quad
  \begin{tproof}
    \tleaf{%
      \tsf{F A:y}%
    }
    \tleaf{
      \tsf{F B:z}%
    }
    \tvrule
    \tinfer[tline]2[\ru{Fstar}]{%
      \tsf{F A star B:x}%
      \tnl%
      \tsc{T \bifL{x}{y}{z}}%
    }
  \end{tproof}
  \quad
  \begin{tproof}
    \tleaf{%
      \tsf{F A:y}%
    }
    \tleaf{
      \tsf{T B:z}%
    }
    \tvrule
    \tinfer[tline]2[\ru{Twand}]{%
      \tsf{T A wand B:x}%
      \tnl%
      \tsc{T \bifL{z}{x}{y}}%
    }
  \end{tproof}
  \quad
  \begin{tproof}
    \tleaf{\tsc{T \bifL{v}{x}{u}}%
           \tnl%
           \tsf{T A:u}\tnl\tsf{F B:v}%
    }
    \tinfer[tline]1[\ru{Fwand}]{\tsf{F A wand B:x}}
  \end{tproof}
  \quad
  \begin{tproof}
    \tleaf{
      \tsf{S A:y}%
    }
    \tinfer[tline]1[\ru{SEqu}]{%
      \tsf{S A:x}%
      \tnl%
      \tsc{T \df{x =L y}}%
    }
  \end{tproof}

  \\ \tabnl \tabheader{Structural rules} \tabnl

  \begin{tproof}
    \tleaf{\tsc{T \leqL{x}{y}}}
    \tleaf{\tsc{T \leqL{y}{x}}}
    \tvrule
    \tleaf{\tsc{T \sepL{x}{y}}}
    \tvrule
    \tinfer[tline]3[\ru{CaseD}]{\tsf{S A:x} \tsp \tsf{S B:y}}
  \end{tproof}
  \quad
  \begin{tproof}
    \tleaf{\tsc{T \leqL{x}{x}}}
    \tinfer[tline]1[\ru{Refl}]{\tsf{S A:x}}
  \end{tproof}
  \quad
  \begin{tproof}
    \tleaf{\tsc{T \leqL{x}{z}}}
    \tinfer[tline]1[\ru{Tran}]{%
      \tsc{T \leqL{x}{y}}%
      \tnl%
      \tsc{T \leqL{y}{z}}%
    }
  \end{tproof}
  \qquad
  \begin{tproof}
    \tleaf{\tsc{T \sepL{xj}{y}}}
    \tinfer[tline]1[\ru{SepP}]{\tsc{T \sepL{x1}{x2}}\tnl\tsc{T \leqL{xi}{y}}}
    \delims{\left[}{\right]}
  \end{tproof}

  \\ \tabnl

  \begin{tproof}
      \tleaf{\tsc{F \leqL{x}{z}}\tsp\tsc{F \leqL{z}{x}}}
      \tinfer[double]1[\ru{Sep}]{\tsc{T \sepL{x}{z}}}
  \end{tproof}
  \qquad
  \begin{tproof}
    \tleaf{%
      \tsc{T \leqL{x}{y}}%
      \tsp%
      \tsc{T \leqL{x}{z}}%
      \tsp%
      \tsc{T \sepL{y}{z}}%
     }
    \tinfer[double]1[\ru{Bif}]{\tsc{T \bifL{x}{y}{z}}}
  \end{tproof}
  \qquad
  \begin{tproof}
    \tleaf{\tsc{T \df{x =L y}}}
    \tinfer[double]1[\ru{Equ}]{%
      \tsc{T \leqL{x}{y}}%
      \tnl%
      \tsc{T \leqL{y}{x}}%
    }
  \end{tproof}
  \hspace{10em}

  \\ \tabnl \tabheader{Closure rules} \tabnl

  \begin{tproof}
    \tleafx
    \tinfer[tline]1[\ru{XForm}]{%
      \tsf{T A:x}%
      \tnl%
      \tsf{F A:x}%
    }
  \end{tproof}
  \qquad
  \begin{tproof}
    \tleafx
    \tinfer[tline]1[\ru{XCons}]{\tsc{T \leqL{x}{y}}\tnl\tsc{F \leqL{x}{y}}}
  \end{tproof}

  \\ \tabnl \tabheader{Derivable rules} \tabnl

  \begin{tproof}
    \tleaf{\tsf{F A:x}\tnl\tsf{F B:x}}
    \tinfer[tline]1[\ru{For}]{\tsf{F A or B:x}}
  \end{tproof}
  \quad
  \begin{tproof}
    \tleaf{\tsf{T A:x}}
    \tleaf{\tsf{T B:x}}
    \tvrule
    \tinfer[tline]2[\ru{Tor}]{\tsf{T A or B:x}}
  \end{tproof}
  \quad
  \begin{tproof}
    \tleaf{\tsf{F A:x}}
    \tleaf{\tsf{T B:x}}
    \tvrule
    \tinfer[tline]2[\ru{Timp}]{\tsf{T A imp B:x}}
  \end{tproof}
  \quad
  \begin{tproof}
    \tleaf{\tsf{T A:x}\tnl\tsf{F B:x}}
    \tinfer[tline]1[\ru{Fimp}]{\tsf{F A imp B:x}}
  \end{tproof}

  \\ \tabnl

  \begin{tproof}
    \tleaf{\tsf{T A:x}\tnl\tsf{T B:x}}
    \tleaf{\tsf{F A:x}\tnl\tsf{F B:x}}
    \tvrule
    \tinfer[tline]2[\ru{Teqv}]{\tsf{T A eqv B:x}}
  \end{tproof}
  \quad
  \begin{tproof}
    \tleaf{\tsf{F A:x}\tnl\tsf{T B:x}}
    \tleaf{\tsf{T A:x}\tnl\tsf{F B:x}}
    \tvrule
    \tinfer[tline]2[\ru{Feqv}]{\tsf{F A eqv B:x}}
  \end{tproof}
  \quad
    \begin{tproof}
    \tleaf{\tsc{T \leqL{x}{u}}\tnl\tsf{T A:u}}
    \tinfer[tline]1[\ru{Tdiam}]{\tsf{T diam A:x}}
  \end{tproof}
  \quad
  \begin{tproof}
    \tleaf{\tsf{F A:y}}%
    \tinfer[tline]1[\ru{Fdiam}]{\tsf{F diam A:x}\tnl\tsc{T \leqL{x}{y}}}
  \end{tproof}
  \quad
  \begin{tproof}
    \tleaf{\tsc{T \sepL{y}{x}}}
    \tinfer[tline]1[\ru{SepC}]{\tsc{T \sepL{x}{y}}}
  \end{tproof}

  \\ \tabnl

  \begin{tproof}
    \tleafx
    \tinfer[tline]1[\ru{Ftrue}]{\tsf{F true:x}}
  \end{tproof}
  \quad
  \begin{tproof}
    \tleafx
    \tinfer[tline]1[\ru{Tfalse}]{\tsf{T false:x}}
  \end{tproof}
  \quad
  \begin{tproof}
    \tleafx
    \tinfer[tline]1[\ru{XSepC}]{\tsc{T \sepL{x1}{x2}}\tnl\tsc{T \leqL{xi}{xj} }}
  \end{tproof}
  \quad
  \begin{tproof}
    \tleafx
    \tinfer[tline]1[\ru{XBifA}]{\tsc{T \bifL{x}{y1}{y2}}\tnl\tsc{F \leqL{x}{yi}}}
  \end{tproof}
  \quad
  \begin{tproof}
    \tleafx
    \tinfer[tline]1[\ru{XBifS}]{\tsc{T \bifL{x}{y1}{y2}}\tnl\tsc{T \leqL{yi}{yj}}}
  \end{tproof}

  \\ \tabnl \tabheader{Side-conditions and comments}

  \parbox{0.9\linewidth}{%
    \begin{itemize}
    \item
    In rules introducing $\df u$ and/or $\df v$, $\df u$ and $\df v$ are distinct labels that are fresh in the branch.
    \item
    In rules involving $i$ and/or $j$, $i \in \{1,2\}$ and $j = 3 - i$.
    \item
    Double lines indicate rules that can be used bottom-up or top-down.
    \item
    The brackets indicate that the rule is optional (included for $\BL$, excluded for lax $\BL$).

    \end{itemize}
  }

  \\ \hline

  \end{tabular}

  \caption{Rules of the $\TBL$ calculus}
  \label{fig:tbl-rules}
\vspace{2mm}
\hrule
\end{figure}}

\begin{definition}\label{def:tbl-proof}
Let $\lf{A:x}$ be a labelled formula.
A \emph{$\TBL$-proof of~$\lf{A:x}$} is a \emph{closed tableau for~$\lf{A:x}$}.
We write $\df{|-TBL \lf{A:x}}$ if $\lf{A:x}$ is \emph{provable~in $\TBL$}, that is if there exists a $\TBL$-proof of~$\lf{A:x}$.
Similarly, a formula~$\df{A}$ is \emph{provable} in~$\TBL$, written $\df{|-TBL A}$, if\/ $\df{|-TBL \lf{A:x}}$ for some label~$\df{x}$.
\fillBox
\end{definition}

{
\begin{figure}[t]
\hrule
\vspace{2mm}
\centering
\begin{tsmallproof}
  \tleafx%
  \tinfer(6)[tline]1[\ru{XForm}]{%
    \tsf{F A:y}(6)%
  }
  \tinfer(5)[tline]1[\ru{Fdiam}]{%
    \tsf{F diam A:x}(5)%
  }
  \tleafx%
  \tinfer(8)[tline]1[\ru{XForm}]{%
    \tsf{F B:z}(8)%
  }
  \tinfer(7)[tline]1[\ru{Fdiam}]{%
    \tsf{F diam B:x}(7)%
  }
  \tvrule(4)
  \tinfer(4)[tline,separation=1em]2[\ru{Fand}]{%
    \tsc{T \leqL{x}{y}}(5)%
    \tsp
    \tsc{T \leqL{x}{z}}(7)%
    \tsp
    \tsc{T \sepL{y}{z}}%
  }
  \tinfer(3)[tline]1[\ru{Bif}]{%
    \tsc{T \bifL{x}{y}{z}}(3)%
    \tnl%
    \tsf{T A:y}(6)%
    \tsp%
    \tsf{T B:z}(8)%
  }
  \tinfer(2)[tline]1[\ru{Tstar}]{%
    \tsf{T A star B:x}(2)%
    \tnl%
    \tsf{F diam A and diam B:x}(4)%
  }
  \tinfer(1)[tline]1[\ru{Fimp}]{%
    \tsf{F A star B imp diam A and diam B:x}
  }
\end{tsmallproof}
  \hspace{4cm}
  \begin{tsmallproof}
  \tleaf{%
    \tsc{T x <=L y1 , z1}%
    \tnl%
    \tsf{T true star true:y1}%
    \tnl%
    \tsf{T true star true:z1}%
  }
  \tinfer(5)[tline]1[\ru{Tboxm}\text{ twice}]{%
    \tsc{T \bifL{x}{y1}{z1}}(5)%
    \tnl%
    \tsf{T true:y1}%
    \tnl%
    \tsf{T true:z1}%
  }
  \tinfer(4)[tline]1[\ru{Tstar}]{%
    \tsf{T true star true:x}(4)%
  }
  \tinfer(3)[tline]1[\ru{Tboxm}]{%
    \tsc{T x <=L x}%
  }
  \tinfer(2)[tline]1[\ru{Refl}]{%
    \tsf{T boxm ( true star true ):x}(2,3,5)%
  }
  \tinfer(1)[tline]1[\ru{Fnot}]{%
    \tsf{F not boxm ( true star true ):x}
  }
\end{tsmallproof}

\caption{Tableaux examples}
\vspace{2mm}
\hrule
\label{fig:two-tableau-examples}
\end{figure}}

\begin{example}\label{ex:tbl-proof-one}
The tableau depicted on the left-hand side of \fig~\ref{fig:two-tableau-examples} is a $\TBL$-proof of $\df{A star B imp diam A and diam B}$.
Step~3 is only given to improve readability and is not really necessary as it is just the explicit expansion of the shorthand for $\sdf{T \bifL{x}{y}{z}}$.  
\fillBox
\end{example}

\begin{example}\label{ex:tbl-proof-two}
A $\TBL$-proof of $\varphi \equiv \df{( diam A star B ) eqv ( A star B )}$ is given in the tableau below. 

\smallskip

\begin{center}{\small
\begin{tsmallproof}
  \tleafx%
  \tinfer(6)[tline]1[\ru{XForm}]{%
    \tsf{F A:y}(6)
  }
  \tinfer(5)[tline]1[\ru{Fdiam}]{%
    \tsc{T \leqL{y}{y}}(5)
  }
  \tinfer(4)[tline]1[\ru{Refl}]{%
    \tsf{F diam A:y}(4,5)
  }
  \tleafx%
  \tinfer(7)[tline]1[\ru{XForm}]{%
    \tsf{F B:z}(7)
  }
  \tvrule(3)
  \tinfer(3)[tline]2[\ru{Fstar}]{%
    \tsc{T \bifL{x}{y}{z}}(3)
    \tnl%
    \tsf{T A:y}(6)
    \tnl%
    \tsf{T B:z}(7)
  }
  \tinfer(2)[tline]1[\ru{Tstar}]{%
    \tsf{F diam A star B:x}(3)
    \tnl%
    \tsf{T A star B:x}(2)
  }
  %
  %
  \tleafx%
  \tinfer(12)[tline]1[\ru{XSepC}]{%
    \tsc{T \sepL{y}{u}}(12)%
  }
  \tinfer(11)[tline]1[\ru{SepP}]{}
  \delims{\left[\!}{\right]}
  \tinfer[no rule]1{%
    \tsc{T \leqL{z}{u}}(11)%
  }
  \tleafx%
  \tinfer(6)[tline]1[\ru{XBifS}]{%
    \tsc{T \leqL{y}{z}}(6)%
  }
  \tinfer(5)[tline]1[\ru{Tran}]{%
    \tsc{T \leqL{u}{z}}(5)%
  }
  \tvrule(4)
  \tleafx%
  \tinfer(9)[tline]1[\ru{XForm}]{%
    \tsf{F A:u}(9)
  }
  \tleafx%
  \tinfer(10)[tline]1[\ru{XForm}]{%
    \tsf{F B:z}(10)
  }
  \tvrule(8)
  \tinfer(8)[tline]2[\ru{Fstar}]{%
    \tsc{T \bifL{x}{u}{z}}(8)%
  }%
  \tinfer(7)[tline]1[\ru{Tran}+\ru{Bif}]{%
    \tsc{T \sepL{u}{z}}(7)%
  }%
  \tvrule(4)
  \tinfer(4)[tline,separation=3pt]3[\ru{CaseD}]{
    \tsc{T \leqL{y}{u}}(5,7,12)
    \tnl
    \tsf{T A:u}(4,9)
  }
  \tinfer(3)[tline]1[\ru{Tdiam}]{%
    \tsc{T \bifL{x}{y}{z}}(6,7,11)
    \tnl
    \tsf{T diam A:y}(3)
    \tnl
    \tsf{T B:z}(4,10)
  }
  \tinfer(2)[tline]1[\ru{Tstar}]{%
    \tsf{T diam A star B:x}(2)
    \tnl%
    \tsf{F A star B:x}(8)
  }
  \tvrule(1)
  \tinfer(1)[tline]2[\ru{Feqv}]{%
    \tsf{F ( diam A star B ) eqv ( A star B ):x}
  }
\end{tsmallproof}
}
\end{center}
After Step~1, the tableau splits into two parts, the interesting one being the second one (on the right-hand side of the first vertical rule) that eventually leads to four closed branches.
Step~4 is an example of the case distinction rule. We remark that the first case only leads to a closed branch when the optional rule $\ru{SepP}$ is used (resulting in the optional steps Step~$11$ and Step~$12$).
Indeed, without persistence of separation, one can build a countermodel of~$\varphi$ as illustrated in \fig~\ref{fig:bl-examples}.
Hence, $\varphi$ is a formula that distinguishes persistent from non-persistent ${\BL}$. 
\fillBox
\end{example}

\smallskip 

\begin{example}\label{ex:tbl-infine}
The tableau depicted on the right-hand side of \fig~\ref{fig:two-tableau-examples} is an example of an infinite open tableau for $\df{not boxm ( true star true )}$.
The signed formula $\slf{T boxm ( true star true ):x}$ introduced in Step~1 must be expanded for all labels~$\df{u}$ such that $\slf{T \leqL{x}{u}}$ occurs in the branch.
The first such label is~$\df{x}$, so that $\slf{T true star true:x}$ is 
introduced in Step~3.
The expansion of $\slf{T true star true:x}$ in Step~4 generates two new successors of~$\df{x}$, namely $\df{y1 , z1}$, and then induces Step~5, where two new expansions of $\slf{T boxm ( true star true ):x}$ are performed (for simplicity, in one step for both $\df{y1}$ and $\df{z1}$).
Step~5 results in the introduction of $\slf{T true star true:y1}$ and $\slf{T true star true:z1}$, the expansion of which 
creates four new successors of~$\df{x}$, two 
of~$\df{y1}$ and two 
of~$\df{z1}$.
The whole process described previously then repeats itself infinitely often as the ${\BL}$ models  such that $\df{x ||- boxm ( true star true )}$ are the ones in which all sucessors of~$\df{x}$ have at least two distinct successors.
This point is further discussed at the beginning of Section~\ref{sec:FMP}.
\fillBox
\end{example}

\section{Soundness and Completeness} \label{sec:SandC}

\begin{definition}\label{def:tbl-domain}
The \emph{domain of a tableau branch~$\tbb$}, denoted $\dom{\tbb}$, is the set
$\SETC{ \df x }{ (\slf{S A:x}) \in \tbb}$ of all labels occurring in~$\tbb$.
\fillBox
\end{definition}

\begin{definition}[Realization]\label{def:realization}
Let $\tbb$ be a tableau branch.
A \emph{realization} of~$\tbb$ in a $\BL$ model $\df{MO = ( Wlds , <=S , RelS , ValS , ||- )}$ is a function $\realFun$ from $\dom{\tbb}$ to $\df{Wlds}$ such that
\begin{itemize}[leftmargin=6mm,label=--]
\setlength\itemsep{-0.5mm}
\item if $( \slf{T A:x} ) \in \tbb$, then $\df{ \real{x} ||- A }$ and
if $( \slf{F A:x} ) \in \tbb$, then $\df{ \real{x} /||- A }$
\item if $( \sdf{T x <=L y} ) \in \tbb$, then $\df{ \real{x} <=S \real{y} }$,
and if $( \sdf{F x <=L y} ) \in \tbb$, then $\df{ \real{x} /<=S \real{y} }$.
\end{itemize}
A tableau branch is \emph{realizable} if has at least one realization in some $\BL$ model.
A tableau is \emph{realizable} if has at least one realizable branch.
\fillBox
\end{definition}

\begin{lemma}\label{lem:closed-branch-not-realizable}
If a tableau branch is closed, then it is not realizable. 
\end{lemma}
\begin{proof}
Let $\tbb$ be a closed tableau branch. Assume that~$\tbb$ is realizable. Then, there exists a realization $\realFun$ of~$\tbb$ in some $\BL$ model. 

\begin{itemize}[leftmargin=6mm,label=--]
\item If $\tbb$ is closed because both $( \slf{T A:x} ) \in \tbb$ and $( \slf{F A:x} ) \in \tbb$, then by definition of a realization, we have both $\df{ \real{x} ||- A }$ and $\df{ \real{x} /||- A}$, which is a contradiction.
\item If $\tbb$ is closed because both $( \sdf{T x <=L y} ) \in \tbb$ and $( \sdf{F x <=L y} ) \in \tbb$, then by definition of a realization, we have both $\df{ \real{x} <=S \real{y} }$ and $\df{ \real{x} /<=S \real{y} }$, which is a contradiction.
\end{itemize}
Therefore, $\tbb$ cannot be realizable.
\end{proof}

\begin{lemma}\label{lem:tbl-preserves-realizabibility}
If a tableau~$\tb$ is realizable, then expanding~$\tb$ using one of the tableau expansion rules given in \fig~\ref{fig:tbl-rules} results in a realizable tableau~$\tb'$. 
\end{lemma}
\begin{proof}
Suppose that~$\tb$ is realizable. Then, it has at least one realizable branch~$\tbb$.
If $\tb'$ is obtained from~$\tb$ by expanding a branch that is distinct from~$\tbb$ then $\tb'$ remains realizable since it still contains the unchanged realizable branch~$\tbb$. Otherwise, $\tb'$ is obtained by expanding~$\tbb$ into~$\tbb'$.
We then proceed by case analysis on the tableau rule expanding~ $\tbb$.
Let $\realFun$ be a realization of~$\tbb$ is some $\BL$ model. 

We consider just a few illustrative cases, the others being similar.
\begin{itemize}[leftmargin=6mm,label=--]
\item Case $\ru{Twand}$:~

If $(\sdf{T \bifL{z}{x}{y}}) \in \tbb$ and $(\sdf{T A wand B:x}) \in \tbb$, then by Definition~\ref{def:realization}, we get $\ter{\real{z}}{\real{x}}{\real{y}}$ and $\df{ \real{x} ||- A wand B}$.
It then follows from Definition~\ref{def:model} that $\df{ \real{y} ||- A }$ and $\df{ \real{z} /||- B}$. Therefore, $\realFun$ realizes~$\tbb'$.

\item Case $\ru{Fwand}$:~

If $(\sdf{F A wand B:x}) \in \tbb$, then by Definition~\ref{def:realization}, we get $\df{ \real{x} /||- A wand B}$. By Definition~\ref{def:model}, there exist $\df{wld1 , wld2 in Wlds}$ such that $\terdf{wld2}{\real{x}}{wld1}$, $\df{wld1 ||- A}$ and $\df{wld2 /||- B}$.
Since~$\tbb'$ extends~$\tbb$ with $(\sdf{T \bifL{v}{x}{u}})$, $(\slf{T A:u})$ and $(\slf{F B:v})$, where $\df{u}$ and $\df{b}$ are distinct fresh labels, we can extend~$\realFun$ into a realization of~$\tbb'$ by setting $\real{u} = \df{wld1}$ and $\real{v} = \df{wld2}$.

\end{itemize}
The other cases are similar. 
\end{proof}

\begin{theorem}[Soundness]\label{thm:tbl-soundness}
If\/ $\df{|-TBL A}$, then $\df{|= A}$. 
\end{theorem}
\begin{proof}
If $\df{|-TBL A}$ then we have closed tableau $\tb$ for $\lf{A:x}$ for some label~$\df{x}$. Assume that $\df{/|= A}$.
Any tableau construction procedure that results in~$\tb$
begins with the tableau $\tb_0$ that consists in the single node
$\slf{F A:x}$.
Since $\tb_0$ is realizable, Lemma~\ref{lem:tbl-preserves-realizabibility} implies that~$\tb$ should also be realizable and should therefore contain at least one realizable branch~$\tbb$. 
By Lemma~\ref{lem:closed-branch-not-realizable}, $\tbb$ cannot be closed. It then follows that $\tb$ is open, which is a contradiction. Thus, we have $\df{|= \, A}$.
\end{proof}

\newcommand{\fulFun}{\Vvdash}
\newcommand{\fulf}[2][\tbb]{(\dmParse{##1;{\slf{##1}}}{,}{#2;})\in#1}
\newcommand{\fulc}[2][\tbb]{(\dmParse{##1;{\sdf{##1}}}{,}{#2;})\in#1}
\newcommand{\ful}[2][\tbb]{#2\in#1}
\newcommand{\nfulf}[2][\tbb]{(\dmParse{##1;{\slf{##1}}}{,}{#2;})\not\in#1}
\newcommand{\nfulc}[2][\tbb]{(\dmParse{##1;{\sdf{##1}}}{,}{#2;})\not\in#1}
\newcommand{\nful}[2][\tbb]{#2\not\in#1}

\begin{definition}\label{def:tbl-saturated}
A tableau branch~$\tbb$ is \emph{saturated} if it satisfies all of the following conditions:
\begin{enumerate}[leftmargin=6.5mm] 
\setlength\itemsep{-0.5mm}
\item \label{cond:ful-reflexivity}
if $\fulf{S A:x}$, then $\fulc{T x <=L x}$
\item \label{cond:ful-case-distinction}
if $\fulf{S A:x ; S B:y}$, then $\fulc{T x <=L y}$ or $\fulc{T y <=L x}$ or $\fulc{T x ><L y}$
\item \label{cond:ful-transitivity}
if $\fulc{T x <=L y}$ and $\fulc{T y <=L z}$, then $\fulc{T x <=L z}$
\item \label{cond:ful-antisym}
if $\fulc{T x =L y}$ and $\fulf{S A:x}$, then $\fulf{S A:y}$
\item \label{cond:ful-bifurcation}
if $\fulc{T x ><L y}$ and $\fulc{T x <=L z}$, then $\fulc{T z ><L y}$

\item \label{cond:ful-T-not}
if $\fulf{T not A:x}$, then $\fulf{F A:x}$
\item \label{cond:ful-F-not}
if $\fulf{F not A:x}$, then $\fulf{T A:x}$
\item \label{cond:ful-T-and}
if $\fulf{T A and B:x}$, then $\fulf{T A:x}$ and $\fulf{T B:x}$
\item \label{cond:ful-F-and}
if $\fulf{F A and B:x}$, then $\fulf{F A:x}$ or $\fulf{F B:x}$

\item \label{cond:ful-T-boxm}
if $\fulf{T boxm A:x}$, 
then $\df{forall. y} \in \dom{\tbb}$,
     if $\fulc{T x <=L y}$, then $\fulf{T A:y}$
\item \label{cond:ful-F-boxm}
if $\fulf{F boxm A:x}$,  
then for some $\df{y} \in \dom{\tbb}$,
     $\fulc{T x <=L y}$ and $\fulf{F A:y}$

\item \label{cond:ful-T-star}
if $\fulf{T A star B:x}$, 
then, 
for some $\df{y , z} \in \dom{\tbb}$, 
       $\fulc{T \bifL{x}{y}{z}}$, $\fulf{T A:y}$, and $\fulf{T B:z}$
\item \label{cond:ful-F-star}
if $\fulf{F A star B:x}$, 
then, 
for all $\df{y , z} \in \dom{\tbb}$, 
       if $\fulc{T \bifL{x}{y}{z}}$, then $\fulf{F A:y}$ or $\fulf{F B:z}$ 

\item \label{cond:ful-T-wand}
if $\fulf{T A wand B:x}$, 
then, 
for all $\df{y , z} \in \dom{\tbb}$, 
       if $\fulc{T \bifL{z}{x}{y}}$, then $\fulf{F A:y}$ or $\fulf{T B:z}$

\item \label{cond:ful-F-wand}
if $\fulf{F A wand B:x}$,\! 
then, 
       for some $\df{y , z} \!\in\! \dom{\tbb}$,\!\! 
       $\fulc{T \bifL{z}{x}{y}}$, $\fulf{T A:y}$,\! and $\fulf{F B:z}$.  $\Box$  
\end{enumerate} 
\end{definition}

\dfAdd{MOB}{\df{MO}_{\tbb}}
\dfAdd{<=B}{\df{<=S}_{\tbb}}
\dfAdd{/<=B}{\df{/<=S}_{\tbb}}
\dfAdd{RelB}{\df{RelS}_{\tbb}}
\dfAdd{WldsB}{\df{Wlds}_{\tbb}}
\dfAdd{]]B}{]_{\tbb}}
\dfAdd{qleB}{\prec_{\tbb}}
\dfAdd{/qleB}{\nprec_{\tbb}}
\dfAdd{eqvB}{\sim_{\tbb}}
\dfAdd{domB}{\dom{\tbb}}
\newcommand{\domB}{\dom{\tbb}}



\begin{lemma}\label{lem:tbl-quasi-order}
Let $\tbb$ be a saturated open tableau branch.
The  binary relation~$\df{qleB}$ over~$\domB$ defined as $\df{x qleB y}$ iff $(\slf{T x <=L y}) \in \tbb$ is a quasi-order over~$\domB$ such that if $\fulf{F x <=L y}$ then $\df{x /qleB y}$.
\end{lemma}
\begin{proof}
Transitivity and reflexivity of~$\df{qleB}$ clearly follow from Conditions~\ref{cond:ful-reflexivity} and~\ref{cond:ful-transitivity} of Definition~\ref{def:tbl-saturated}. Now, if $\fulc{F x <=L y}$, then $(\slf{T x <=L y}) \not\in \tbb$ because~$\tbb$ is open. 
Hence, $\df{x /qleB y}$ by definition of~$\df{qleB}$.
\end{proof}

\begin{lemma}\label{lem:tbl-partial-order}
Let $\tbb$ be a saturated open branch.
Let~$\df{eqvB}$ be the equivalence relation over~$\domB$ defined as $\df{x eqvB y}$ iff $\df{x qleB y}$ and $\df{y qleB x}$.
For all $\df{x in domB}$, let $\df{ecx} = \SETC{ \df{y in domB} }{ \df{x eqvB y} }$ denote the equivalence class of~$\df{x}$ under~$\df{eqvB}$ and let $\df{WldsB}$ denote the quotient $\domB / \df{eqvB}$.
Then, the binary relation~$\df{<=B}$ over~$\df{WldsB}$ defined as $\df{ecx <=B ecy}$ iff $\df{x qleB y}$ is a partial order over~$\df{WldsB}$ that satisfies the persistent separation property.
\end{lemma}

\begin{proof}
First, we show that $\df{<=B}$ is well defined by showing that the following conditions are equivalent for all $\df{ecu , ecv in WldsB}$:
\begin{itemize}
\item[$(a)$] $\df{u qleB v}$
\item[$(b)$] $\df{x qleB y}$ for all $\df{x in ecu}$ and all $\df{y in ecv}$
\item[$(c)$] $\df{x qleB y}$ for some $\df{x in ecu}$ and some $\df{y in ecv}$
\end{itemize}
It is clear that~$(a)$ implies~$(c)$ and that~$(b)$ implies~$(c)$ (and $(a)$).
Therefore, we only have to show that~$(c)$ implies~$(b)$.
Assume $\df{x qleB y}$ for some $\df{x in ecu}$ and some $\df{y in ecv}$.
Then, $\df{x eqvB u}$ implies $\df{u qleB x}$ and $\df{y eqvB v}$ implies $\df{y qleB v}$.
Pick an arbitrary $\df{x' in ecu}$, then $\df{x' eqvB u}$ implies $\df{x' qleB u}$.
Since $\df{u qleB x}$, we get $\df{x' qleB x}$.
Pick an arbitrary $\df{y' in ecv}$, then $\df{y' eqvB v}$ implies $\df{v qleB y'}$
Since $\df{y qleB v}$, we get $\df{y qleB y'}$.
Finally, since we assumed $\df{x qleB y}$, 
$\df{x' qleB x}$ and $\df{y qleB y'}$ imply $\df{x' qleB y'}$.

Second, we show that $\df{<=B}$ is a partial order over~$\df{WldsB}$.
Since $\df{qleB}$ is a quasi-order over~$\domB$ by Lemma~\ref{lem:tbl-quasi-order}, it immediately follows that~$\df{<=B}$ is both reflexive and transitive.
It remains to show that~$\df{<=B}$ is anti-symmetric.
Assume $\df{ecx <=B ecy}$ and $\df{ecy <=B ecx}$, then, by definition of~$\df{<=B}$, we have $\df{x qleB y}$ and $\df{y qleB x}$, from which we get
$\df{x eqvB y}$ by definition of~$\df{eqvB}$.
Hence, $\df{ecx = ecy}$.

\dfAdd{><B}{\df{><S}}

Last, we show that~$\df{<=B}$ satisfies the separation persistence property stated in Definition~\ref{def:tree}.
We pick arbitrary $\df{ecx , ecy , ecz in WldsB}$ such that $\df{ecx ><S ecy}$ and $\df{ecx <=B ecz}$ and show
that $\df{ecz ><S ecy}$.
We have the following facts:
\begin{itemize}
\item[$(i)$] By definition of~$\df{><S}$, $\df{ecx ><S ecy}$ implies $\df{ecx /<=B ecy}$ and $\df{ecy /<=B ecx}$.
\item[$(ii)$] By definition of $\df{<=B}$, $\df{ecx <=B ecz}$ implies $\df{x qleB z}$.
\end{itemize}
Assume that $\df{ecz <=B ecy}$.
Then, $\df{z qleB y}$ by definition of~$\df{<=B}$. 
Since $\df{x qleB z}$ and $\df{z qleB y}$ imply $\df{x qleB y}$, we get $\df{ecx <=B ecy}$ by definition of~$\df{<=B}$, which contradicts $\df{ecx ><S ecy}$ $(i)$.
Hence, $\df{ecz /<=B ecy}$.~\\
Assume that $\df{ecy <=B ecz}$.
Then, $\df{y qleB z}$ by definition of~$\df{<=B}$, from which we get $\fulc{T y <=L z}$ by definition of~$\df{qleB}$.
Besides, $\df{ecx /<=B ecy}$ and $\df{ecy /<=B ecx}$ imply $\df{x /qleB y}$ and $\df{y /qleB x}$ by definition of~$\df{<=B}$.
By definition of~$\df{qleB}$, we get $\nfulc{T x <=L y}$ and $\nfulc{T y <=L x}$.
Since~$\tbb$ is saturated, $\fulc{T x ><L y}$ then follows from Condition~\ref{cond:ful-case-distinction} of Definition~\ref{def:tbl-saturated}. 
In turn, $\fulc{T x ><L y}$ and $\fulc{T y <=L z}$ imply $\fulc{T x ><L z}$ by Condition~\ref{cond:ful-bifurcation} of Definition~\ref{def:tbl-saturated}.
$\fulc{T x ><L z}$ implies $\fulc{F x <=L z}$ by definition of~$\df{><L}$,
but from $\fulc{F x <=L z}$, Lemma~\ref{lem:tbl-quasi-order} implies $\df{x /qleB z}$, which contradicts $\df{x qleB z}$ $(ii)$.
Hence, $\df{ecy /<=B ecz}$.
From, $\df{ecz /<=B ecy}$ and $\df{ecy /<=B ecz}$, we conclude $\df{ecy ><S ecz}$ by definition of~$\df{><S}$.
\end{proof}

\begin{lemma}[Model existence]\label{lem:tbl-model-existence}
Let $\tbb$ be a saturated open tableau branch.
Then, $\tbb$ induces a $\BL$ model  $\df{ MOB = ( Wlds , <=S , RelS , ValS , ||- ) }$,  where $\df{Wlds = WldsB}$ and $\df{<=S \, = \, <=B}$ as per Lemma~\ref{lem:tbl-partial-order}, $\df{RelS}$ is the ternary relation induced by~$\df{<=B}$ as per Definition~\ref{def:frame}, and, for all worlds $\df{ecx in WldsB}$, $\ValS{ecx} = \SETC{ \df{p} }{ ( \slf{T p:x} ) \in \tbb }$.
Moreover, $\df{MOB}$ is such that if $\fulf{T A:x}$ then $\df{ecx ||- A}$ and if $\fulf{F A:x}$ then $\df{ecx /||- A}$. 
\end{lemma}
\begin{proof}
By mutual induction on the structure of~$\df{A}$ for all labels~$\df{x}$.
\begin{description}

\item[Base case $\df{A = p}$:]~

If $\fulf{T p:x}$, then $\df{p} \in \ValS{ecx}$ by definition of $\df{ValS}$, which implies $\df{ecx ||- p}$ by Definition~\ref{def:model}.

\smallskip

If $\fulf{F p:x}$, then since~$\tbb$ is open, we have $(\slf{T p:x}) \not\in \tbb$. 
Hence, we get $\df{p /in \ValS{ ecx }}$ by definition of $\df{ValS}$, which implies $\df{ecx /||- p}$ by Definition~\ref{def:model}.

\item[Case $\df{A = A1 wand A2}$:]~

If $\fulf{F A1 wand A2:x}$, then for some $\df{y , z} \in \dom{\tbb}$, we have $\fulc{T \bifL{z}{x}{y}}$, $\fulf{T A1:y}$ and $\fulf{F A2:z}$.
$\fulc{T \bifL{z}{x}{y}}$ implies $\terdf{ecz}{ecx}{ecy}$ by Lemma~\ref{lem:tbl-partial-order}.
Moreover, by induction hypothesis, $\fulf{T A1:y}$ and $\fulf{F A2:z}$ imply
$\df{ecy ||- A1}$ and $\df{ecz /||- A2}$.
Hence, by Definition~\ref{def:model}, we have $\df{ecx /||- A1 wand A2}$.

\smallskip

If $\fulf{T A1 wand A2:x}$, then for all $\df{y , z} \in \dom{\tbb}$, if $\fulc{T \bifL{z}{x}{y}}$ then $\fulf{F A1:y}$ or $\fulf{T A2:z}$.
Pick arbitrary $\df{ecu , ecv in WldsB}$ such that $\terdf{ecv}{ecx}{ecu}$ and $\df{ecu ||- A1}$.
Since $\df{ecv <=B ecx}$ and $\df{ecv <=B ecu}$, we have $\fulc{T v <=L x}$ and $\fulc{T v <=L u}$ by definition of~$\df{<=B}$.
Similarly, since $\df{ecx /<=B ecu}$ and $\df{ecu /<=B ecx}$, we have $\nfulc{T x <=L u}$ and $\nfulc{T u <=L x}$, which by Condition~\ref{cond:ful-case-distinction} implies $\fulc{T x ><L u}$.
It then follows that we have $\fulc{T \bifL{v}{x}{u}}$, which implies
$\fulf{F A1:u}$ or $\fulf{T A2:v}$.
By the induction hypothesis, we then get $\df{ecu /||- A1}$ or $\df{ecv ||- A2}$.
Since we assume $\df{ecu ||- A1}$, we necessarily have $\df{ecv ||- A2}$.
Therefore, $\df{ecx ||- A1 wand A2}$.  
\end{description}
The other cases are similar. 
\end{proof}

\begin{corollary}\label{cor:counter-model}
If~$\tbb$ is a saturated open tableau branch in a tableau
for $\lf{A:x}$, then the induced model $\df{ MOB = ( WldsB , <=B , RelS , ValS , ||- ) }$ is such that $\df{ecx /||- A}$. 
\end{corollary}
\begin{proof}
Since $( \slf{F A:x} )$ is the root of any tableau for $\lf{A:x}$, The concluding property  
of Lemma~\ref{lem:tbl-model-existence} implies that $\df{x /||- A}$.    
\end{proof}

\begin{theorem}[Completeness]\label{thm:tbl-completeness}
If\/ $\df{||- A}$, then $\df{|-TBL A}$. 
\end{theorem}
\begin{proof}
It is standard to define a tableau construction procedure that applies the rules given in \fig~\ref{fig:tbl-rules} with a fair strategy.
Such a procedure will either result in a finite closed tableau for~$\df A$, in which case we get a $\TBL$-proof of~$\df A$, or build at least one (possibly infinite) complete open branch~$\tbb$, in which case $\tbb$ gives rise to a $\BL$ model $\df{MOB}$ such that $\df{ MOB /|= A}$ by Corollary~\ref{cor:counter-model}.
\end{proof}

\section{A Finite Model Property and Decidability} \label{sec:FMP}

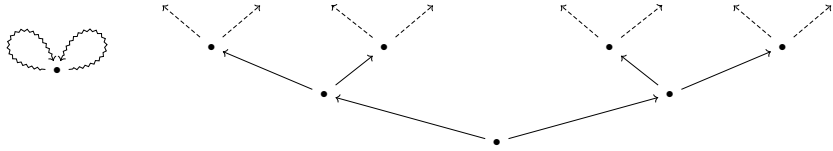
\begin{figure}
\hrule 
\vspace{2mm}
    \centering
    {\tiny  
\begin{tikzcd}
	\bullet
	\arrow[squiggly, from=1-1, to=1-1, loop, in=110, out=175, distance=10mm]
	\arrow[squiggly, from=1-1, to=1-1, loop, in=70, out=5, distance=10mm]
\end{tikzcd}
 \begin{tikzcd}[ampersand replacement=\&]
	{} \&\& {} \& {} \&\& {} \&\& {} \&\& {} \& {} \&\& {} \\
	\& \bullet \&\&\& \bullet \&\&\&\& \bullet \&\&\& \bullet \\
	\&\&\& \bullet \&\&\&\&\&\& \bullet \\
	\&\&\&\&\&\& \bullet
	\arrow[dashed, from=2-2, to=1-1]
	\arrow[dashed, from=2-2, to=1-3]
	\arrow[dashed, from=2-5, to=1-4]
	\arrow[dashed, from=2-5, to=1-6]
	\arrow[dashed, from=2-9, to=1-8]
	\arrow[dashed, from=2-9, to=1-10]
	\arrow[dashed, from=2-12, to=1-11]
	\arrow[dashed, from=2-12, to=1-13]
	\arrow[from=3-4, to=2-2]
	\arrow[from=3-4, to=2-5]
	\arrow[from=3-10, to=2-9]
	\arrow[from=3-10, to=2-12]
	\arrow[from=4-7, to=3-4]
	\arrow[from=4-7, to=3-10]
\end{tikzcd}}

    \caption{A finite model with backlinks (left) and its infinite unfolding (right)} 
    \vspace{1mm}
    \hrule
    \label{fig:backlinks}
\end{figure}

The formula $\df{boxm ( true star true )}$ enforces that every world has at least two distinct successors and can, therefore, only be satisfied in infinite models like the one shown in Figure~\ref{fig:backlinks} (right). Consequently, the finite model property in its traditional formulation fails for $BL$.

However, models such as the one in Figure~\ref{fig:backlinks} (right) are sufficiently regular to allow for a finite schematic representation. Consider, for example, the scheme in Figure~\ref{fig:backlinks} (left). Here, each squiggly edge represents a \emph{backlink}, indicating that a copy of the graph originating from its endpoint should be attached to its starting point via a single edge. By unfolding such a schematic model, we obtain precisely the full binary tree shown in Figure~\ref{fig:backlinks} (right).

 For simplicity, we only present an informal description of \emph{models with backlinks}. These models are trees augmented with additional backlinks --- edges that point backward in the tree order. Formulas in BL can be interpreted on such models by first unfolding them into (infinite) trees. One can then establish the following result:

%


\begin{lemma} \label{lem:backlinks}
    If $\varphi$ is $\df{wand}$-free and satisfiable in $\BL$, then it has a finite model with backlinks. 
\end{lemma}

\begin{proof}
 Let $\Sigma$ be the set of subformulae of $\varphi$. For a subset $\Sigma_0\subseteq\Sigma$, consider all immmediate successors of the root node that satisfy exactly those formulae in $\Sigma$ that are in $\Sigma_0$. If we remove all but two of such nodes and the branches stemming from them we obtain a different model that still satisfies $\varphi$ at the root. Note that we need to retain two nodes--unlike in standard modal logic where one representative suffices--as the root might satisfy a statement $\psi\ast\psi$ that enforces the existence of two \emph{different} successors, even with equal theories. Applying this idea repeatedly we can create a model of $\varphi$ that is finitely branching up to any given height.

On the other hand, if a branch is longer than $2^{|\Sigma|}$, we will encounter worlds $w<v$ satisfying exactly the same formulae from $\Sigma$. In this case, we can remove the branch stemming from $v$ and instead create a back link from the immediate predecessor of $v$ to $w$. This too preserves truth of $\varphi$ at the root.
In the end, we obtain a finitely branching model with backlinks that is of bounded height, and therefore finite.   
\end{proof}

The proof is a modification of a standard filtration and pruning argument used in other modal logics \cite{BdeRV2001}. These arguments rely on the fact that the truth of a formula depends only on the truth of its subformulae at successor nodes --- a property that fails once one considers backwards-looking modalities. This explains the restriction on $\df{wand}$.

\begin{corollary}
    It is decidable whether a $\df{wand}$-free formula is valid in $\df{BL}$. 
\end{corollary}
\begin{proof}
By Theorem~\ref{thm:tbl-completeness},
     the $\df{wand}$-free formulae of $\df{BL}$ are computably enumerable. Furthermore, Lemma~\ref{lem:backlinks} implies that  and it follows 
     that their complement is computably enumerable too, since we can systematically generate all finite models with backlinks. Thus, decidability follows.
 \end{proof}

\subsection*{Acknowledgements} 

\begin{enumerate} [leftmargin=6.5mm] 
\setlength\itemsep{-0.5mm}
\item Galmiche and M\'ery gratefully acknowledge the partial support of the ANR Projet NARCO (ANR-21-CE48-0011).

\item Lang and Pym gratefully  acknowledge the partial support of the UK EPSRC through Research Grants EP/S013008/1 and  EP/R006865/1. 

\item We thank the anonymous referees for their suggestions. 
\end{enumerate}

\bibliographystyle{eptcs}
\bibliography{mybib}

\end{document}